\documentclass[conference,letterpaper]{IEEEtran}

\usepackage[utf8]{inputenc} 
\usepackage[T1]{fontenc}
\usepackage{url}
\usepackage{ifthen}
\usepackage{cite}
\usepackage[cmex10]{amsmath} % Use the [cmex10] option to ensure complicance
\usepackage{amssymb,lipsum} 
\usepackage{ bm}
\allowdisplaybreaks
\usepackage{amsthm}            % with IEEE Xplore (see bare_conf.tex)
\usepackage{tikz}
\usepackage{graphicx}
\usepackage{url}
\usepackage{pgfplots}
\usepackage{caption}
\usepackage{subfig}
\usepackage{float, stfloats}
\usepgfplotslibrary{external} 
\usepackage{latexsym}

\newcommand{\cN}{{\mathcal{N}}}

\newcommand{\bY}{\bm{Y}}
\newcommand{\bZ}{\bm{Z}}
\newcommand{\bA}{\bm{A}}
\newcommand{\bU}{\bm{U}}

\newcommand{\bB}{\bm{B}}

\newcommand{\bx}{\bm{x}}

\newcommand{\bE}{\mathbb{E}}

\newcommand{\bs}{\mathbf{s}}

\newcommand{\bR}{\mathbb{R}}
\newcommand{\bP}{\mathbf{P}}
\DeclareMathOperator{\Tr}{Tr}

\newtheorem{theorem}{Theorem}
\newtheorem{remark}{Remark}
\newtheorem{lemma}{\textbf{Lemma}}
\usepackage{amsthm}

\pgfplotsset{compat=1.17}
\begin{document}
\title{Mismatched Estimation of Rank-One Symmetric Matrices Under Gaussian Noise} 

% %%% Single author, or several authors with same affiliation:
% \author{%
%  \IEEEauthorblockN{Stefan M.~Moser}
%  \IEEEauthorblockA{ETH Zürich\\
%           ISI (D-ITET)\\
%           CH-8092 Zürich, Switzerland\\
%           Email: moser@isi.ee.ethz.ch}
% }

%%% Several authors with up to three affiliations:
\author{\IEEEauthorblockN{Farzad Pourkamali and Nicolas Macris}
\IEEEauthorblockA{LTHC, EPFL, Lausanne, Switzerland}
Emails: \{farzad.pourkamali,nicolas.macris\}@epfl.ch}

\maketitle

%%%%%%
%% Abstract: 
%%
\begin{abstract}
We consider the estimation of an $n$-dimensional vectors from noisy element-wise measurements of $\bs \bs^T$, a problem that frequently arises in statistics and machine learning. We investigate a mismatched Bayesian inference setting in which the statistician is unaware of some of the parameters. For the particular case of Gaussian priors and additive noise, we derive the complete exact analytic expression for the asymptotic mean squared error (MSE) in the large system size limit. Our formulas demonstrate that estimation is still possible in the mismatched case and that the statistician can achieve the minimum MSE (MMSE) by selecting appropriate parameters. Our technique is based on the asymptotic behavior of spherical integrals and can be used as long as the statistician chooses a rotationally invariant prior. 
\end{abstract}

\section{Introduction}
Many problems in machine learning and statistics can be expressed as estimating a low-rank matrix from its noisy observation. Examples are sparse PCA \cite{DM14}, the spiked Wigner model, community detection \cite{DAM15}. For the rank-one symmetric case, the problem is formulated as follows: a vector $\bs \in \bR^n$ is generated with i.i.d. elements distributed according to $s_i \sim \bP^*$, the matrix $\bs \bs^T$ is observed through an element-wise additive white gaussian noise channel. The goal is to estimate the vector $\bs$ upon observing the noisy version of $\bs\bs^T$.

The statistical and computational limits of this problem have been extensively studied. Most works have so far considered the "Bayes-optimal" setting, in which the prior $P$ and possibly other hyper-parameters (e.g., SNR) are known to the statistician. In the Bayes-optimal setting, computing the mutual information enables us to compute the minimum mean squared error (MMSE) and derive the information-theoretical limits of the estimation. The analytical but highly non-rigorous replica and cavity methods rooted in statistical physics have been used to derive expressions for the mutual information between the true signal and the observation matrix \cite{LKL15}. These expressions were already rigorously derived in early work \cite{KoradaMacris2009} for binary signals using Guerra-Toninelli interpolation \cite{GuerraToninelli2002}. Later the problem has been studied in much detail for general signals, in \cite{BDM16} used approximate message passing (AMP) and spatial coupling, \cite{LM17} Guerra-Toninelli interpolation and Aizenman-Sims-Starr methods, and \cite{BM19}, \cite{BarbierMacris2019} the adaptive interpolation method, to rigorously prove the limiting expressions of mutual information and MMSE. All these methods crucially rely on the assumption that the prior and the parameters of the estimation problem are known to the statistician. The Bayes law then induces remarkable identities that enable the analysis to proceed. In the present case, we lack such identities.
%In the statistical physics language, the system is said to be \textit{on the Nishimori line} \cite{N01}. 

Despite the vast amount of work on this problem in the Bayes-optimal setting, to the best of our knowledge, there is no rigorous result for the {\it mismatched} case corresponding to the realistic situation where the statistician does not know the true prior or/and hyper-parameters, and can only make assumptions about them. Mismatched inference for the scalar and vector estimation problems has been considered in \cite{V10}, \cite{T10}. In particular, \cite{V10} proved a result relating the MSE in the mismatched inference to the relative entropy of the true prior and the statistician's prior. We follow this work and define the MSE similarly (up to natural modification for the matrix case). 

The main contribution of this paper is to compute the full asymptotic of the mismatched MSE for the matrix factorization problem in the large $n$ limit. Our approach uses the results on the spherical integrals from the mathematical physics literature \cite{GM05}. A primary assumption in our method that would be difficult to dispense of is the rotational invariance of the statistician's prior. Despite this restriction, we can study non-rotation invariant true priors, non-symmetric matrix estimation, higher-ranks (finite w.r.t $n\to +\infty$). In this short note, we limit ourselves to the theoretical limits of mismatched estimation for the case of Gaussian priors (both for the true and the fake one) and postpone the detailed study of the more general cases to a forthcoming detailed work. As will become clear in section \ref{secMainResults}, already under this limited setting, the phase transitions phenomenology is quite rich.

The rest of the paper is organized as follows. In Section II, we introduce the setting and formulate the problem. Section III describes the main result and discusses it in several special cases, followed by the proof sketch of the main theorem in Section IV. Lastly, we conclude the paper with some remarks and possible future directions for this line of work.

\section{Problem Setting}
Suppose $\bs \in \bR^n$ is generated with i.i.d. elements from $P^*=\cN(0, \sigma^2)$, the observed matrix is
\begin{equation}
 \bY = \sqrt{\frac{\lambda}{n}} \bs \bs^T + \bZ
 \label{Spiked-Wigner}
\end{equation}
where $\lambda$ is the signal-to-noise-ratio (SNR), and the noise matrix $\bZ$ is a symmetric matrix with i.i.d. $\cN(0, 1)$ off-diagonal and $\cN(0, 2)$ diagonal entries. This model is called the \textit{Spiked-Wigner model}. The purpose of the scaling factor $\frac{1}{\sqrt{n}}$ is to make the inference problem neither trivially easy nor completely impossible in the large system limit.

The statistician is aware that the channel is additive Gaussian and that the true prior is a centered Gaussian, but he does not know the values $\lambda$ and $\sigma$.
He assumes values $\lambda'$ and $\sigma'$ as the SNR and the prior variance. Following the Bayesian estimation principle, he chooses the posterior mean as the estimate for the ground-truth. Our goal is to compute the asymptotic of the MSE for this mismatched estimation problem. Define the \textit{mismatched} matrix-MSE as
\begin{equation}
 \text{MSE}_n(\sigma, \sigma', \lambda, \lambda') := \frac{1}{n^2} \bE_{\bP^*, \bP_{\bZ}}\Big[\big\|\bs \bs^T - \langle \bx \bx^T \rangle_{\lambda',\sigma'}\big\|_F^2 \Big]
\end{equation}
where $\|.\|_F$ is the Frobenius norm, and $\langle . \rangle_{\lambda',\sigma'}$ denotes the expectation with respect to the posterior distribution from the statistician's point of view, that the SNR is $\lambda'$ and $\bx \sim \bP=\cN(0, \sigma'^2)$. Here we adopt the traditional statistical mechanics notation for the internal (annealed) expectations 
$$
\langle f(\bx) \rangle_{\lambda',\sigma'} = \frac{\int d\bx\, \bP(\bx) f(\bx) e^{-\frac{1}{4} \Vert \sqrt{\frac{\lambda}{n}} \bs\bs^T + \bZ - \sqrt{\frac{\lambda'}{n}} \bx\bx^T \Vert_F^2}}{\int d\bx\, \bP(\bx) e^{-\frac{1}{4} \Vert \sqrt{\frac{\lambda}{n}} \bs\bs^T + \bZ - \sqrt{\frac{\lambda'}{n}} \bx\bx^T \Vert_F^2}}
$$
for any reasonable function $f(\bx)$ such that the integrals are finite.

Note that, when we are in the matched (Bayes optimal) case $\lambda'=\lambda$, $\sigma'=\sigma$, the best achievable error is the matrix-MMSE which is defined as
\begin{equation}
 \text{MMSE}_n(\sigma, \lambda) := \frac{1}{n^2} \bE_{\bP^*, \bP_{\bZ}}\Big[\big\|\bs \bs^T - \langle \bx \bx^T \rangle_{\lambda,\sigma}\big\|_F^2 \Big]
\end{equation}
We necessarily have $\text{MSE}_n \geq \text{MMSE}_n$.

\section{Main Result}\label{secMainResults}
The main result is the following:
\begin{theorem}
Assume that the sequence $({\rm MSE})_{n\geq 1}$ converges uniformly in $(\lambda, \lambda') \in K \subset \mathbb{R}_+^2 $, then for 
For all $\sigma, \sigma'$ (strictly positive) and $(\lambda, \lambda')\in K$, the asymptotic mismatched MSE is given by eq. (\ref{MSE_general}).
\begin{figure*}[b!]
\centering
\begin{minipage}{\textwidth}
\vspace{-0.1cm}
\hrule
\begin{equation}
 \lim_{n \rightarrow \infty} \text{MSE}_n(\sigma, \sigma', \lambda, \lambda') = \left\{
\begin{array}{ll}
  \sigma^4 + \big(\frac{1}{\sqrt{\lambda'}}-\frac{1}{\lambda' \sigma'^2}\big)^2 & \text{if } \lambda \leq \frac{1}{\sigma^4}\text{, and } \lambda' \geq \frac{1}{\sigma'^4} \\
  \sigma^4(1-\sqrt{\frac{\lambda}{\lambda'}})^2 + \frac{2}{\sqrt{\lambda \lambda'}} + \frac{1}{\lambda'^2 \sigma'^4} + \frac{2}{\lambda'} \frac{\sigma^2}{\sigma'^2} (1 - \sqrt{\frac{\lambda}{\lambda'}}) - \frac{2}{\lambda \lambda' \sigma^2 \sigma'^2}& \text{if } \lambda \geq \frac{1}{\sigma^4}\text{, and } \sqrt{\lambda \lambda'} \geq \frac{1}{\sigma^2 \sigma'^2} \\
  \sigma^4 & \text{if o.w.}
\end{array}
\right.
\label{MSE_general}
\end{equation}
\medskip
\end{minipage}
\end{figure*}
\end{theorem}

\begin{remark}\label{remark-on-unif-conv}
In the matched case, uniform convergence of the sequence $({\rm MMSE})_{n\geq 1}$ - except possibly at phase transition points which form a set of measure zero - follows using the concavity of mutual information with respect to $\lambda$. Then, using the I-MMSE relation \cite{GuoVerduMMSE2005}, this allows to interchange limit and derivative to go from asymptotic mutual information (a.k.a. free energy) to asymptotic MMSE. For the present mismatched MSE, we use a relation similar to I-MMSE but in terms of mismatched free energies, which lack concavity w.r.t. $\lambda$ and $\lambda'$. Therefore almost everywhere, uniform convergence is difficult to establish from general principles. However, we conjecture that it holds and that eq. (\ref{MSE_general}) holds almost everywhere (i.e., except possibly at phase transition lines).
\end{remark}

The MSE is illustrated for the case of $\sigma=1, \lambda=2$ in Fig. 1. The observed behavior is generic for $\lambda\sigma^4 >1$. We observe one phase transition line and an intermediate region where estimation better than chance is possible, in the sense that the MSE is smaller than $\sigma^4$. We refer to the caption of Fig. 1 for details. In the case $\sigma =1$ and $\lambda <1$, or more generally $\lambda\sigma^4 < 1$, it is easy to see from Eq. \eqref{MSE_general} that the intermediate region disappears and the MSE is always greater or equal to $\sigma^4$ (the phase transition line is still present technically speaking).

\begin{figure}[H]
% \vspace{-1.5cm}
  \centering
  %subfloat[]
  {
\includegraphics[scale = 0.7]{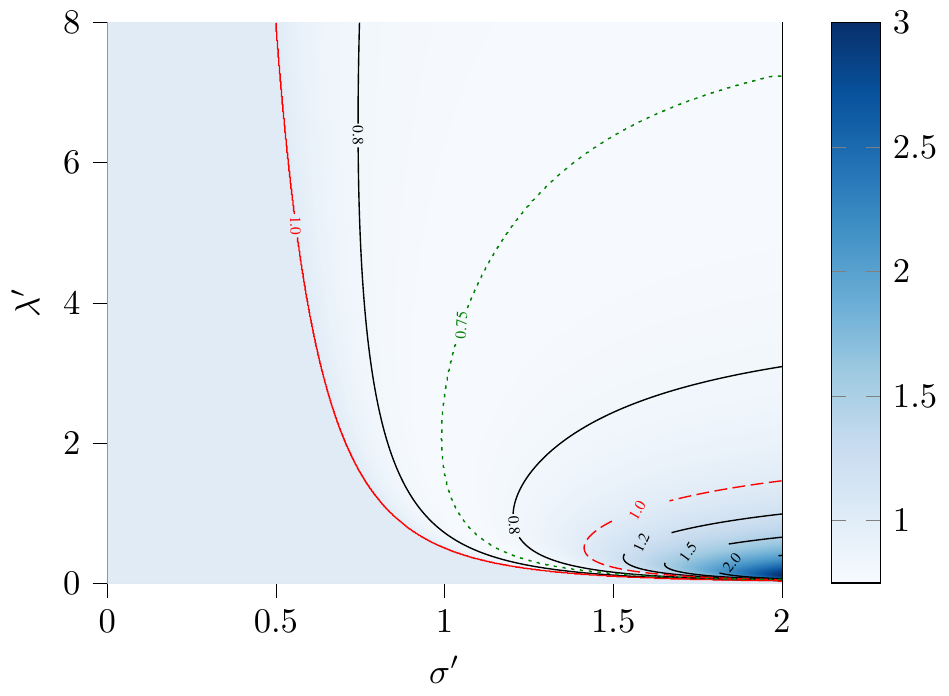}
}
\caption{\small Plot of MSE according to Eq. \eqref{MSE_general} for $\sigma = 1, \lambda = 2$. The solid leftmost (red) curve is a {\it phase transition} line. On the left of this curve ${\rm MSE = \sigma^4 =1}$. In the {\it intermediate region} between the solid leftmost (red) curve and the dashed (red) curve the MSE takes values less than $\sigma^4 =1$. In this intermediate region estimation better than chance is possible. On the dotted (green) curve the MSE attains the  ${\rm MMSE}(\sigma, \lambda) = \frac{2}{\lambda} - \frac{1}{\lambda^2 \sigma^4}=0.75$ (even though we do not have $\lambda'=\lambda$, $\sigma'=\sigma$ except for one point with a vertical tangent on the curve). The MSE equals 
$\sigma^4=1$ on the dashed (red) line and takes higher values in the region on the right hand side of this line. Note that this is {\it not} a phase transition line. Finally we point out that the MSE is continuous throughout and the phase transition is therefore a {\it continuous phase transition}.
The analytical expressions of the phase transition line, as well as dotted and dashed lines can easily be written down from eqs. \eqref{MSE_general} and \eqref{eqmmse}. For $\sigma = 1, \lambda = 2$ the dotted (resp. dashed) curves have horizontal asymptotes $\lambda'=8$ (resp. $\lambda'=2$).}
% \vspace{-0.1cm}
\end{figure}

Figures 2 and 3 depict the behavior of the MSE along vertical and horizontal sections of Fig. 1. We
clearly observe that the MSE is not monotonous and that for $\lambda' < 8$, the minimal value given by the MMSE may be achieved. 
These observations can be checked analytically from the expressions of the MSE and MMSE.

 \begin{figure}[H]
  \centering
{
\includegraphics[scale = 0.8]{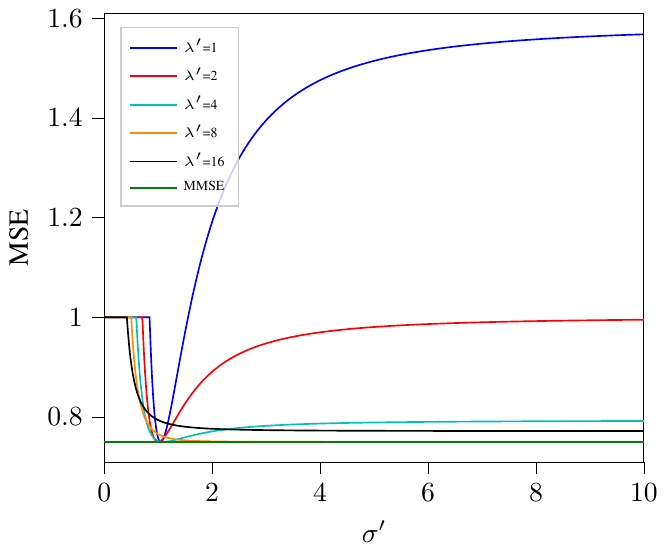}
}
\caption{\small Behavior of the MSE as a function of $\sigma'$. Here $\sigma = 1, \lambda = 2$. The horizontal (green) level gives the value of the ${\rm MMSE} = \frac{2}{\lambda} - \frac{1}{\lambda^2\sigma^4} = 0.75$ in the matched case. We have
$\lim_{\sigma'\to +\infty}{\rm MSE} = 1+\frac{2}{\lambda'} - \sqrt{\frac{2}{\lambda'}}$ and this limiting value is decreasing (resp. increasing) for $\lambda' < 8$ (resp. $\lambda' > 8$). For $\lambda' > 2$ estimation better then chance is possible for large enough $\sigma'$.}
\end{figure}
\begin{figure}[H]
\centering
% \subfloat[]{
{
\includegraphics[scale = 0.8]{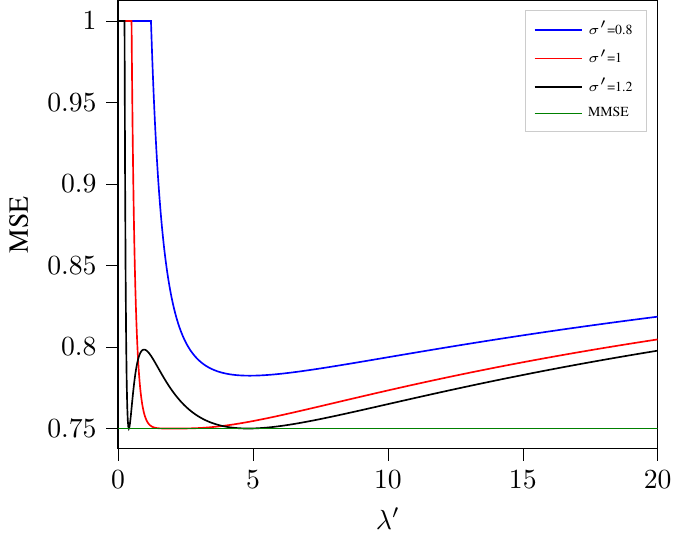}
}
\caption{\small Behavior of the MSE as a function of $\lambda'$. Here $\sigma = 1, \lambda = 2$.
The horizontal (green) level gives the value of the ${\rm MMSE} = \frac{2}{\lambda} - \frac{1}{\lambda^2\sigma^4} = 0.75$ in the matched case.
All curves have horizontal asymptote $\sigma^4 =1$ for $\lambda' \to+\infty$.}
\end{figure}

\subsection{Inference with Matched SNR}
Suppose that the statistician fully knows the channel and can choose $\lambda' = \lambda$. The mismatched MSE then reduces to:
\begin{align}
\begin{aligned}
&
\text{if } \sigma' \leq \sigma, \, \lim_{n \rightarrow \infty} \text{MSE}_n(\sigma, \sigma', \lambda, \lambda) = 
\\ 
&
\left\{
\begin{array}{ll}
  \sigma^4 & \text{if } \lambda \leq \frac{1}{\sigma^2 \sigma'^2}\\
  \frac{2}{\lambda} - \frac{1}{\lambda^2 \sigma'^2}(\frac{2}{\sigma^2} - \frac{1}{\sigma'^2}) & \text{if } \lambda \geq \frac{1}{\sigma^2 \sigma'^2}
\end{array}
\right.
\nonumber
%\label{MSE_s'<s}
\\
\\
&
 \text{if } \sigma' \geq \sigma, \, \lim_{n \rightarrow \infty} \text{MSE}_n(\sigma, \sigma', \lambda, \lambda) = 
 \\ 
 &
 \left\{
\begin{array}{ll}
  \sigma^4 & \text{if } \lambda \leq \frac{1}{ \sigma'^4}\\
  \sigma^4 + \frac{1}{\lambda} - \frac{1}{\lambda^{\frac{3}{2}} \sigma'^2}(2 - \frac{1}{\sqrt{\lambda} \sigma'^2}) & \text{if } \frac{1}{\sigma'^4} \leq \lambda \leq \frac{1}{\sigma^4} \\
  \frac{2}{\lambda} - \frac{1}{\lambda^2 \sigma'^2}(\frac{2}{\sigma^2} - \frac{1}{\sigma'^2}) & \text{if } \lambda \geq \frac{1}{\sigma^4}
\end{array}
\right.
%\label{MSE_s'>s}
\nonumber
\end{aligned}
\end{align}

For $\sigma = 1$ the MSE is plotted as a function of SNR for various values of $\sigma'$ in Fig. 4. When $\sigma' > \sigma$, we observe that the MSE increases as the SNR increases (a similar behavior occurs on Fig. 1 in \cite{V10} for the scalar case). Although this happens when we are still in the regime of small SNR and estimation is impossible, we find this behavior rather counterintuitive.

\begin{figure}[H]
\centering
\includegraphics[scale = 0.8]{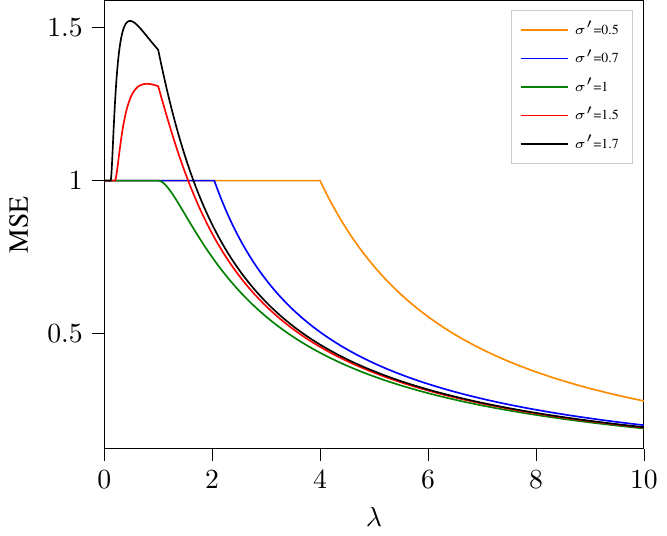}
\caption{\small Behavior of the MSE for matched SNR $\lambda' = \lambda$.}
\end{figure}

\begin{remark}
For $\sigma' = \sigma$, in fact we are in the Bayes optimal setting and we can find the minimum MSE (MMSE). In this case, previous equations reduce to
\begin{equation}\label{eqmmse}
 \lim_{n \rightarrow \infty} {\rm MMSE}_n(\sigma, \lambda) = \left\{
\begin{array}{ll}
  \sigma^4 & {\rm if } \lambda \leq \frac{1}{\sigma^4}\\
  \frac{2}{\lambda} - \frac{1}{\lambda^2 \sigma^4} & {\rm if } \lambda \geq \frac{1}{\sigma^4}
\end{array}
\right.
\end{equation}
This expression is well known and derived previously by a host of different approaches (see \cite{DM14, DAM15, LM17, BM19, BDM16}).
\end{remark}

As a sanity check of our result for the matched SNR case, with a bit of work we can check explicitly that
\begin{equation}
\begin{split}
\int_0^{\infty} [\text{MSE}(\sigma, \sigma', \lambda, \lambda) - &\text{MMSE}(\sigma, \lambda)] \, d \lambda \\
&= 4 D_{KL}(\cN(0,\sigma^2), \cN(0,\sigma'^2))
\end{split}
\end{equation}
where $D_{KL}$ denotes the Kullback-Leibler divergence. This sum-rule for vector channels is derived in \cite{V10} (with a factor of $2$ instead of $4$ in the vector case).

\section{Analysis}
\subsection{Mismatched free Energy and MSE}
From the statistician's point of view, the posterior distribution reads up to a normalizing factor
\begin{equation}
 \begin{split}
   \bP\{ \bx | \bY\} &\propto e^{-\frac{1}{4} \|\bY - \sqrt{\frac{\lambda'}{n}}\bx \bx^T \|_F^2}\bP(\bx) \\
 & \propto e^{-\frac{\lambda'}{4n}\|\bx\|^4 + \frac{1}{2}\sqrt{\frac{\lambda'}{n}} \Tr \bY \bx \bx^T}\bP(\bx)
 \end{split}
 \label{Post_dist}
\end{equation}
where $\bP$ is the normal distribution with iid entries and variance $\sigma'$. In deriving the second line, we use the fact that $\|\bY\|_F$ is a constant (because it is being conditioned on). Note that,  $\bY$ is symmetric and the upper (or lower) part is distributed as $(Y_{i,j})_{i<j} \sim \cN(\sqrt{\frac{\lambda}{n}}s_is_j,1)$, and the diagonal $(Y_{i,i}) \sim \cN(\sqrt{\frac{\lambda}{n}}s_i s_i,2)$.

The \textit{partition function} is defined as the normalization factor of the last expression
\begin{equation}
 Z(\bY) = \int d\bx \, e^{-\frac{\lambda'}{4n}\|\bx\|^4 + \frac{1}{2}\sqrt{\frac{\lambda'}{n}} \Tr \bY \bx \bx^T}\bP(\bx)
 \label{partition_func}
\end{equation}
and the {\it mismatched free energy} is defined as 
\begin{equation}
 f_n(\sigma, \sigma', \lambda, \lambda') = -\frac{1}{n} \bE_{\bP^*, \bP_{\bZ}} [\ln Z(\bY)]
 \label{free_energy}
\end{equation}
Now we state a lemma relating the mismatched free energy to MSE. Keep in mind that both mismatched free energy and MSE are functions of $\sigma, \sigma',\lambda, \lambda'$, but for simplicity of notation, we drop the arguments. 

\begin{lemma}
\begin{equation}
  \frac{d}{d \lambda'} f_n+(2-\sqrt{\frac{\lambda}{\lambda'}})\sqrt{\frac{\lambda}{\lambda'}}\frac{d}{d \lambda} f_n + \frac{1}{4} \sigma^4 = \frac{1}{4} {\rm MSE}_n
\label{f-mse}
\end{equation}
\end{lemma}

\begin{remark}
Eq. \eqref{f-mse} generalizes the classical I-MMSE relation. Here the mismatched free energy cannot be related to a mutual information. However, note that, in the special case where $\lambda' = \lambda$ Eq. \eqref{f-mse} simplifies slightly and combining with the I-MMSE relation, we obtain that the difference of MSE and MMSE is directly related to a derivative of a relative entropy, equivalent to relations discussed in detail in \cite{V10} for vector channels. 
\end{remark}

\textit{Proof of lemma.}
We have
\begin{equation}
 \frac{d}{d \lambda} f_n = - \frac{1}{4} \frac{1}{n^2} \sqrt{\frac{\lambda'}{\lambda}} \bE_{\bP^*, \bP_{\bZ}} \Big[ \big\langle (\bs^T \bx)^2 \big\rangle_{\lambda',\sigma'} \Big]
\end{equation}
and by using a standard Gaussian integration by parts trick,
\begin{equation}
 \frac{d}{d \lambda'} f_n = \frac{1}{4} \frac{1}{n^2} \bE_{\bP^*, \bP_{\bZ}} \Big[ \big\| \langle \bx \bx^T \rangle_{\lambda',\sigma'} \big\|_F^2 - \sqrt{\frac{\lambda}{\lambda'}} \big\langle (\bs^T \bx)^2 \big\rangle_{\lambda',\sigma'} \Big]
\end{equation}
Putting these two equations together, the left-hand side of eq. (\ref{f-mse}) is equal to
\begin{align}
  &\frac{1}{4} \frac{1}{n^2} \bE_{\bP^*, \bP_{\bZ}} \Big[ \big\| \langle \bx \bx^T \rangle_{\lambda',\sigma'} \big\|_F^2 - 2 \big\langle (\bs^T \bx)^2 \big\rangle_{\lambda',\sigma'} + \|\bs\|^4 \Big] 
  \nonumber \\
  &= \frac{1}{4} \frac{1}{n^2} \bE_{\bP^*, \bP_{\bZ}} \Big[ \big\| \langle \bx \bx^T \rangle_{\lambda',\sigma'} \big\|_F^2 - 2 \Tr \bs \bs^T \langle \bx \bx^T \rangle_{\lambda',\sigma'}
  \nonumber \\ &
  \hspace{4mm}+ \|\bs \bs^T\|_F^2 \Big]
= \frac{1}{4} \text{MSE}_n \hspace{0.7cm} \square
\end{align}

Thus, the problem is reduced to computing the (mismatched) free energy.
The main idea is to exploit the rotational invariance of the normal distribution. Changing variables $\bx \rightarrow \bU \bx$, for an orthogonal matrix $\bU \in \bR^{n \times n}$, the integral in eq. (\ref{partition_func}) becomes ($|{\rm det} \bU| =1$):
\begin{equation}
 \begin{split}
   Z(\bY) &= \int d\bx \, e^{-\frac{\lambda'}{4n}\|\bU \bx\|^4 + \frac{1}{2}\sqrt{\frac{\lambda'}{n}} \Tr \bY \bU \bx \bx^T \bU^T}\bP( \bU \bx) \\
 &= \int d\bx \bP(\bx)\, e^{-\frac{\lambda'}{4n}\|\bx\|^4 + \frac{1}{2}\sqrt{\frac{\lambda'}{n}} \Tr \bY \bU \bx \bx^T \bU^T}
 \end{split}
\end{equation}
Since this holds for any orthogonal matrix $\bU$, we can take the expectation over the \textit{Haar} measure on the group of $n \times n$ orthogonal matrices.
\begin{equation}
 Z(\bY) = \int d\bx \bP(\bx)\, e^{\frac{-\lambda'}{4n}\|\bx\|^4} \int D \bU e^{\frac{1}{2}\sqrt{\frac{\lambda'}{n}} \Tr \bY \bU \bx \bx^T \bU^T}
 \label{part_spher}
\end{equation}
where $D \bU$ denotes the \textit{Haar} measure.

In the next subsection, we will discuss how to compute the inner integral in eq. (\ref{part_spher}).

\subsection{Spherical Integrals}
The spherical integral is defined as:
\begin{equation}
 I_n(\bA,\bB) = \int D \bU e^{n \Tr \bA \bU \bB \bU^T }
 \label{Spherical_int}
\end{equation}
where $\bA, \bB \in \bR^{n \times n}$, and $D \bU$ denotes the \textit{Haar} measure over the orthogonal matrices. Note that, this definition can also be extended to the unitary matrices. In the mathematical physics literature, such integrals are often called \textit{Harish-Chandra-Itzykson-Zuber (HCIZ)} integrals. The interest for these objects dates to the work of the mathematician Harish-Chandra \cite{HC}, and they have been extensively studied and developed in physics and mathematics. 
In particular, \cite{GM05} estimated the asymptotics of spherical integrals when the rank of matrix $\bB$ is $O(1)$ w.r.t $n$. We will apply this result to our case.

From the definition (\ref{Spherical_int}), one may notice that the integral only depends on the eigenvalues of $\bA$,$\bB$. So, it is natural to expect that the asymptotic of the integral depends on the limiting spectral measure of the matrix $\bA$. The result of \cite{GM05} is based on the hypothesis that the spectral measure $\mu_{\bA}$ converges weakly towards a compactly supported measure $\mu$, and the minimum and maximum eigenvalues of $\bA$ converge to the finite values $\gamma_{\text{min}}$, $\gamma_{\text{max}}$, respectively.

For a probability measure $\mu$, the \textit{Hilbert} (or \textit{Stieltjes}) transform is the map $H_{\mu}:\bR \backslash \text{supp}(\mu) \rightarrow \bR$,
$H_{\mu}(z) = \int \frac{1}{z - t} \, d \mu(t)$.
This map is invertible, and denoting its inverse by $H_{\mu}^{-1}(.)$, for $z$ in range of $H_{\mu}$ we define the \textit{R-transform} of a probability measure $\mu$ as $R_{\mu}(z) = H_{\mu}^{-1}(z) - \frac{1}{z}$.

%Now, we are ready to state Theorem 6 from \cite{GM05}:
\begin{theorem}[Guionnet and Maida \cite{GM05}]
Suppose $\mu_{\bA}$ converges weakly towards $\mu$ and ${\rm rank}(B) =1$. Let $H_{\text{min}} = \lim_{ z \rightarrow \gamma_{\text{min}}} H_{\mu}(z)$, $H_{\text{max}} = \lim_{ z \rightarrow \gamma_{\text{max}}} H_{\mu}(z)$, and $\theta$ be the non-zero eigenvalue of $\bB$, then:
\begin{equation}
 \begin{split}
  \lim_{n \rightarrow \infty} \frac{1}{n} \ln & I_n(\bA, \bB) \\
  &= \theta \nu(\theta) - \frac{1}{2} \int \ln(1+ 2 \theta \nu(\theta) - 2 \theta t) \, d \mu(t)
 \end{split}
\end{equation}
where
\begin{equation}
 \nu ( \theta ) = \left\{
\begin{array}{ll}
  R_{\mu} ( 2 \theta ) & \text{if } H_{\text{min}} \leq 2 \theta \leq H_{\text{max}} \\
  \gamma_{\text{max}} - \frac{1}{2\theta} & \text{if } 2 \theta > H_{\text{max}} \\
  \gamma_{\text{min}} - \frac{1}{2\theta} & \text{if } 2 \theta < H_{\text{min}}
\end{array}
\right.
\end{equation}
\end{theorem}

\subsection{Computing Free Energy}
To apply the result from \cite{GM05}, we can rewrite the spherical integral in eq. (\ref{part_spher}) as
\begin{equation}
 \int D \bU e^{n \Tr \frac{\bY}{\sqrt{n}} \bU \frac{\sqrt{\lambda'}}{2n}\bx \bx^T \bU^T}
 \label{spher_part}
\end{equation}

$\frac{\bY}{\sqrt{n}} = \frac{\sqrt{\lambda}}{n} \bs \bs^T + \frac{1}{\sqrt{n}} \bZ$, where $\frac{1}{\sqrt{n}} \bZ$ is the suitably normalized Wigner matrix whose limiting spectral measure is the renowned \textit{semi-circle law} with density $\mu_{\rm SC}=\frac{1}{2 \pi} \sqrt{4-t^2} dt$. At the same time, the spectral measure of $\frac{\bY}{\sqrt{n}}$ converges almost surely (a.s) as $n \rightarrow \infty$ to the \textit{semi-circle law} (see e.g. proposition 1 in \cite{CD16}). We have $H_{\mu_{SC}}(z) = \frac{1}{2}(z - \sqrt{z^2 - 4})$ and $R_{\mu_{SC}}(z) = z$.

Let $\gamma_{\text{min}}$ and $\gamma_{\text{max}}$ be the bottom and top eigenvalue of $\frac{\bY}{\sqrt{n}}$, from the results in \cite{BN11}, we have (a.s.)
\begin{equation}
  \gamma_{\text{min}} = -2, \gamma_{\text{max}} = \left\{
\begin{array}{ll}
  2 & \text{if } \frac{\sqrt{\lambda} \|\bs\|^2}{n} \leq 1\\
  \frac{\sqrt{\lambda} \|\bs\|^2}{n} + \frac{n}{\sqrt{\lambda} \|\bs\|^2}& \text{if } \frac{\sqrt{\lambda} \|\bs\|^2}{n} \geq 1
\end{array}
\right.
\end{equation}
So,
\begin{equation}
  H_{\text{min}} = -1, H_{\text{max}} = \left\{
\begin{array}{ll}
  1 & \text{if } \frac{\sqrt{\lambda} \|\bs\|^2}{n} \leq 1\\
  \frac{n}{\sqrt{\lambda} \|\bs\|^2}& \text{if } \frac{\sqrt{\lambda} \|\bs\|^2}{n} \geq 1
\end{array}
\right.
\end{equation}

On the other hand the non-zero eigenvalue of the rank-one matrix $\frac{\sqrt{\lambda'}}{2n}\bx \bx^T$ is $\theta = \frac{\sqrt{\lambda'}}{2n}\|\bx\|^2$. Thus, the asymptotic of the integral in eq. \eqref{spher_part} is only a function of $\|\bx\|$ and $\|\bs\|$, and can be computed by theorem 2 for the different cases of the parameters.

\begin{theorem}
For all $\sigma, \sigma', \lambda, \lambda'$ positive, the asymptotic of free energy for the mismatched inference is given in eq. (\ref{f_general}).
\begin{figure*}[t]
\centering
\begin{minipage}{\textwidth}
\begin{equation}
 \lim_{n \rightarrow \infty} f_n(\sigma, \sigma', \lambda, \lambda') = \left\{
\begin{array}{ll}
  -\frac{1}{4 \lambda' \sigma'^4} + \frac{1}{\sqrt{\lambda'} \sigma'^2} - \frac{3}{4} + \ln \lambda'^{\frac{1}{4}}\sigma' & \text{if } \lambda \leq \frac{1}{\sigma^4}\text{, and } \lambda' \geq \frac{1}{\sigma'^4} \\
  \frac{1}{2} \ln \sqrt{\lambda \lambda'} \sigma^2 \sigma'^2 -\frac{1}{4 \lambda' \sigma'^4} - \frac{\lambda \sigma^4}{4} + \sqrt{\frac{\lambda}{\lambda'}}\frac{\sigma^2}{2 \sigma'^2} + \frac{1}{2 \sqrt{\lambda \lambda'} \sigma^2 \sigma'^2 } - \frac{1}{2} & \text{if } \lambda \geq \frac{1}{\sigma^4}\text{, and } \sqrt{\lambda \lambda'} \geq \frac{1}{\sigma^2 \sigma'^2} \\
  0 & \text{if o.w.}
\end{array}
\right.
\label{f_general}
\end{equation}
\medskip
\vspace{-0.4cm}
\hrule
\end{minipage}
\end{figure*}
\end{theorem}
\begin{proof}[Proof sketch]
We have
\begin{equation}
 \begin{split}
    f_n &= -\frac{1}{n} \bE_{\bP^*, \bP_{\bZ}} [\ln Z(\bY)]\\
 &= -\frac{1}{n} \int d\bs \bP^*(\bs) \bE_{\bP_{\bZ}} \Big[\ln \int d\bx \bP(\bx) \, e^{\frac{-\lambda'}{4n}\|\bx\|^4 + \ln I_n} \Big]
 \end{split}
 \label{f_proof}
\end{equation}
where 
\begin{equation}
 I_n = \int D \bU e^{n \Tr \frac{\bY}{\sqrt{n}} \bU \frac{\sqrt{\lambda'}}{2n}\bx \bx^T \bU^T}
\end{equation}
It is not difficult to see that $I_n$ is invariant under the transformation $\bx \to \bm{R} \bx$ where $\bm{R}$ is a rotation matrix. Therefore the integrand in the $\bx$-integral in \eqref{f_proof} is a function of $\Vert \bx\Vert$. Furthermore recalling $\bY = \sqrt{\frac{\lambda}{n}} \bs \bs^T + \bZ$ and using rotation invariance of $\bP_{\bZ}$ we see that the integrand of the $\bs$-integral is a function of $\Vert \bs\Vert$. Therefore we can use spherical coordinates to reduce the integrals in \eqref{f_proof} to two one-dimensional integrals which yields
\begin{equation}
 \begin{split}
   & f_n = - \frac{2^{-\frac{n}{2}+1}}{\Gamma(\frac{n}{2})}\frac{1}{\sigma^n} \int_0^{+\infty} dr \, r^{n-1} e^{-\frac{r^2}{2 \sigma^2}} \times \\
   & \mathbb{E}_{\bP_{\bZ}}\bigg[\frac{1}{n} \ln \Big\{ \frac{2^{-\frac{n}{2}+1}}{\Gamma(\frac{n}{2})}\frac{1}{\sigma'^n}\, \int_0^{+\infty} d \rho \, \rho^{n-1} e^{-\frac{\rho^2}{2\sigma'^2} - \frac{-\lambda'}{4n}\rho^4 + \ln I_n} \Big\}\bigg]
 \end{split}
\end{equation}
where $r := \|\bs\|$, $\rho := \| \bx\|$, and $\Gamma(.)$ is the \textit{Gamma} function.

Changing variable $\frac{r^2}{n} \rightarrow r$, $\frac{\rho^2}{n} \rightarrow \rho$, we obtain
\begin{equation}
 f_n = - \frac{2^{-\frac{n}{2}}n^{\frac{n}{2}}}{\Gamma(\frac{n}{2})}\frac{1}{\sigma^n} \int_0^{+\infty} \frac{d r}{r} \, e^{-n (\frac{r}{2 \sigma^2} - \frac{1}{2} \ln r)} \mathbb{E}_{\bP_{\bZ}} [g_n(r)]
 \label{f_proof_2}
\end{equation}
where $g_n(r)$ is
\begin{equation}
\frac{1}{n} \ln \Big\{ \frac{2^{-\frac{n}{2}}n^{\frac{n}{2}}}{\Gamma(\frac{n}{2})}\frac{1}{\sigma'^n}\, \int_0^{+\infty} \frac{d \rho}{\rho} \, e^{-n( \frac{\lambda'}{4}\rho^2 + \frac{\rho}{2 \sigma'^2} + \frac{1}{2} \ln \rho - \frac{1}{n} \ln I_n)} \Big\}
\end{equation}
Carefully applying the Laplace method (using Theorem 2.1 in \cite{O74}) and considering different cases for the asymptotics of $\frac{1}{n} \ln I_n$ yields the asymptotics of $g_n(r)$. The result is independent of $\bZ$, and finally, the asymptotics of $f_n$ in (\ref{f_proof_2}) can be computed using the Laplace method again.
\end{proof}

Once we have the expression for the free energy, we can compute the MSE using Lemma 1. As explained in remark \ref{remark-on-unif-conv} this step uses the assumption that for $(\lambda, \lambda')\in K \subset \mathbb{R}_{+}^2$ the sequence $({\rm MSE})_{n\geq 1}$ converges uniformly. 
%\iffalse
%\begin{remark}
%By introducing the temperature parameter $\beta$ in (\ref{Post_dist}) as $\bP\{ \bx | \bY\} \propto e^{\frac{-1}{4} \beta \|\bY - \sqrt{\frac{\lambda'}{n}}\bx \bx^T \|_F^2}\bP^{\beta}(\bx)$, we could consider a larger class of distributions, however this will effect the asymptotic of the free energy by a rescaling factor of parameters, 
%\end{remark}
%\fi

\section{Conclusion}
Studying inference problems in settings where priors and hyper-parameters are unknown or partially known and deriving fundamental limits of estimation is a problem with practical importance. We derived analytical formulas for asymptotic MSE in estimating a rank-one matrix corrupted by additive Gaussian noise when both the channel and prior are partially known. In this short note, we have shown how to treat one of the most straightforward such situations by using beautiful asymptotic formulas of spherical integrals. The major limitation of our technique is that the statistician assumes a spherically invariant prior. This can be a Gaussian which has the advantage of being factorized, but we can also treat a uniform distribution over a sphere. Given such distributions for the statistician, it is then possible to extend our analysis to a broader class of problems, namely:
\begin{itemize}
\item
Estimation of finite rank matrices can be accomodated (i.e., ${\rm rank} = O(1)$ w.r.t $n\to +\infty$).
\item
The true prior does not need to be rotation invariant. General factorized priors can be accommodated, for example, a Rademacher-Bernoulli mixture modeling sparse signals. 
\item
A temperature parameter can be introduced by the statistician in his mismatched posterior distribution (with minor modifications on the analysis).
\end{itemize}
These extensions result in a very rich phenomenology with many possible phase transitions. Already in the simplest situation considered here, the MSE displays non-trivial features. 
Other problems of interest are the construction of more general estimators (non-Bayesian or non-Gibbsian) which can still be analyzed through spherical integrals, as well as confronting the analytical expressions of the MSE to algorithmic predictions, for example, those based on AMP \cite{RF12}, or Approximate Survey Propagation \cite{LFsurvey} applied to mismatched situations.

%From the algorithmic point of view, one interesting direction is to study the performance of the existing algorithms, particularly approximate message passing (AMP) \cite{RF12}, on this problem. Moreover, designing an estimator for the parameters from the observation matrix ( or in general observed data) and studying the MSE of the estimation for the predicted parameters is an intriguing direction.
\section*{Acknowledgment}
\vspace{-0.1cm}
The work of F. P has been supported by the Swiss National Science Foundation grant no 200021E 175541.
N.M is thankful to Jean Barbier and Emanuele Mingione for numerous discussions.
%\texttt{IEEEtran.cls}. 


\end{document}